\documentclass[leqno,12pt]{article}  

\usepackage[margin=1.0in]{geometry}
\usepackage{amsfonts}
\usepackage{amsmath}
\usepackage[noend]{algpseudocode}
\def\R{\mathbb{R}}
\def\L{\mathcal{L}}
\def\C{\mathcal{C}}
\def\A{\mathcal{A}}
\def\orig{\mathbf{0}}
\newtheorem{theorem}{Theorem}
\newtheorem{lemma}{Lemma}
\newtheorem{proposition}{Proposition}
\newtheorem{corollary}{Corollary}
\newenvironment{proof}{\unskip{\bf Proof:}}{\unskip{\hfill $\Box$}}
\def\qed{\unskip{\hfill $\Box$}}

\title{The Complexity of Order Type Isomorphism}
\date{}
\author{Greg Aloupis\thanks{Charg\'e de recherches du F.R.S.-FNRS,
    D\'{e}partement d'Informatique, Universit\'{e} Libre de Bruxelles, 
{\tt aloupis.greg@gmail.com}}
 \and 
 John Iacono\thanks{Department of Computer Science and Engineering, Polytechnic Institute of New York University, {\tt jiacono@poly.edu, ozgurozkan@gmail.com}. Research partially supported by NSF grants CCF-0430849,  CCF-1018370.}
  \and 
 Stefan Langerman\thanks{Directeur de Recherches du
   F.R.S.-FNRS,~D\'{e}partement d'Informatique, Universit\'{e} Libre de Bruxelles,
   {\tt stefan.langerman@ulb.ac.be}. Research supported by F.R.S.-FNRS
 and DIMACS.}
  \and
   \"{O}zg\"{u}r \"{O}zkan\footnotemark[2]
 \and 
    Stefanie Wuhrer\thanks{Cluster of Excellence MMCI, Saarland University,
   {\tt swuhrer@mmci.uni-saarland.de} }}

\begin{document}

\maketitle

\begin{abstract}\small\baselineskip=9pt 
The order type of a point set in $\R^d$ maps each $(d{+}1)$-tuple of points
to its orientation (e.g., clockwise or counterclockwise in $\R^2$). 
Two point sets $X$ and $Y$ have the same order type if there
exists a mapping $f$ from $X$ to $Y$ for which every $(d{+}1)$-tuple 
$(a_1,a_2,\ldots,a_{d+1})$ of $X$ and the corresponding
tuple $(f(a_1),f(a_2),\ldots,f(a_{d+1}))$ in $Y$ have the same orientation.
In this paper we investigate the complexity of determining whether two
point sets have the same order type.
We provide an $O(n^d)$ algorithm for this task, thereby improving upon
the $O(n^{\lfloor{3d/2}\rfloor})$ algorithm of Goodman and Pollack (1983). The
algorithm uses only order type queries and also works for abstract
order types (or acyclic oriented matroids).
Our algorithm is optimal, both in the abstract setting and for
realizable points sets if the algorithm only uses order type queries.
\end{abstract}

\section{Introduction}
In the design of geometric algorithms, as well as in their practical
implementation, it is often convenient to encapsulate the geometry of
a given
problem into a small set of elementary geometric predicates. A typical
example, ubiquitous in computational geometry textbooks, is the \emph{left
turn / right turn} determinant
$$\nabla(a,b,c) =
\begin{vmatrix}
a_x & a_y & 1 \\
b_x & b_y & 1 \\
c_x & c_y & 1 \\
\end{vmatrix}
$$
whose sign ($>0$, $<0$, or $0$) determines if three points $a, b, c\in\R^2$
are in clockwise or counterclockwise orientation, or collinear, respectively.

The practical motivation for this encapsulation will be obvious to any
programmer: by restricting the use of arithmetic operations to just
one place in the code, it is easier to control the robustness of the
code (e.g. with respect to roundoff errors). It is also easier to
generalize the code should a different geometric space require a
slightly different implementation of the predicate (e.g. solving 
geometric problems on a sphere or in a polygon). This
would require a proper abstraction to generalize the predicate
$\nabla$ to other applications.

The need for a classification or discretization of planar point sets
became evident long before computers were invented. 
In 1882,
Perrin~\cite{Perrin1882} described how a point moving on a line far enough from a
collection of points sees the points under a sequence of $n\choose 2$
different radial orders, each produced by swapping two adjacent labels
from the previous ordering. He then showed how this representation
can be used to solve problems without the use of the original
point set. This view of point configurations was revived and
characterized under the name of \emph{allowable sequences} 
by Goodman and Pollack in 1980~\cite{Goodman1980220}. They later showed how the same
allowable sequences can describe pseudoline arrangements~\cite{Goodman1984257}.

The classification of point sets induced by the determinant $\nabla$ above, but
generalized to $d$ dimensions, was discovered around the same time.
Consider a set $P = \{p_1,\ldots,p_n\}\subseteq\R^d$, let $p_i = (x_{i,1},\ldots,x_{i,d})$ and for ease of notation let $x_{i,0}=1$.
The \emph{order type} of $P$ is characterized  by the predicate\footnote{We write $[n]$ to denote the set of integers $\{1,\ldots,n\}$.}
\begin{multline*}
\nabla^P(i_0,i_1,\ldots,i_d) = \nabla(p_{i_0},\ldots,p_{i_d})\\
= \text{sign}(\det(p_{i_0},\ldots,p_{i_d})),
\text{ for all }\{i_0,\ldots,i_d\}\in [n].
\end{multline*}
This concept appeared independently in various contexts over a span of 15
years, under various names, e.g., 
\emph{$n$-ordered sets}~\cite{Novoa}, \emph{multiplex}~\cite{multiplex},
\emph{chirotope}~\cite{dreiding},
\emph{order type}~\cite{goodman_pollack_83_sorting}, among others inspired by problems from chemistry.
For some of them (e.g., chirotopes or abstract order types), 
the precise algebraic definition above is replaced by a set of axioms that the
predicate $\nabla^P$ must satisfy. 

In the early 90's Knuth~\cite{DBLP:books/sp/Knuth92} revisited once more the axiomatic system of
chirotopes under the name of \emph{CC-systems}, but this time with a
specific focus on computational aspects, mainly, what predicates and
axioms are necessary in order to compute a convex hull (and later a
Delaunay triangulation), and what running times can be obtained by an
algorithm using only those predicates.

The theory of \emph{oriented matroids}
 appeared in the mid
'70s. Their primary purpose was to provide an abstraction of linear
dependency.  However, through their various equivalent axiomatizations
they have been used to show a translation between virtually all
abstractions mentioned above. 

\paragraph{Isomorphism}
is probably one of the most fundamental problems for any
discrete structure. In graph theory,
determining whether two graphs are isomorphic is one of the few
standard problems in NP not yet known to be NP-complete and not known
to be in P.
In our setting, two (abstract) order types with predicates $\nabla^P$
and $\nabla^Q$ are \emph{identical}\footnote{Sometimes, two order types are
  also considered identical if all orientations are reversed, i.e.,
  $\nabla^P=-\nabla^Q$.} 
if 
$$\nabla^P(i_0,\ldots,i_d) = \nabla^Q(i_0,\ldots,i_d)
\text{ for all }\{i_0,\ldots,i_d\}\subseteq [n],$$
or more succinctly $\nabla^P=\nabla^Q$.
They are \emph{isomorphic} if there is a permutation $\pi$ such
that 
\begin{multline*}
\nabla^P(i_0,\ldots,i_d) = \nabla^Q(\pi(i_0),\ldots,\pi(i_d))\\
\text{ for all }\{i_0,\ldots,i_d\}\subseteq [n],$$ 
\end{multline*}
or more succinctly, $\nabla^P=\nabla^Q\circ\pi$.

In their seminal paper~\cite{goodman_pollack_83_sorting}, Goodman and Pollack listed an extensive array
of applications of order type isomorphism.  One of these was to be able
to efficiently list all point set configurations in order to test 
several important conjectures in discrete geometry, such as the Erd\H{o}s
and Szekeres conjecture on convex independent sets. In 2002,
Aichholzer et al.~\cite{aichholzer_etal_02} took on that challenge and
proceeded to build a database of  order types for up to 10 points,
which was later extended to 11
points~\cite{aichholzer_krasser_05_abstract_order_type}. Using this
database, they were able to provide new bounds for several open
problems.

Given a labeled point set $P$, an \emph{Order Type Representation} (OTR) is
a function $E$ that only depends on $\nabla^P$ and  encodes the
order type as a string, meaning that using that string, the orientation 
$\nabla^P(i_0,\ldots,i_d)$ of every $d{+}1$-tuple can be retrieved.
We will write $E(P) = E(\nabla^P)$ for that string.
For example, in 1983, Goodman and Pollack~\cite{goodman_pollack_83_sorting} implicitly defined an
encoding of size $O(n^d)$ which lists for every $d$-tuple of integers
$(i_1,\ldots,i_d)$ the number of values $i_0$ for which 
$\nabla^P(i_0,\ldots,i_d)=+$. They showed that  these values suffice to retrieve
the value of $\nabla^P$ for every $d{+}1$-tuple. 

One strategy for identifying whether $P$ and $Q$  have the same order
type is  to fix a labeling of $P$, try every possible
labeling for $Q$, and compare their OTRs, that is, to check whether 
$E(\nabla^P)=E(\nabla^Q\circ\pi)$ for any permutation $\pi$.
In~\cite{goodman_pollack_83_sorting} 
it was shown
 that for comparing two order
types, it suffices to look at a reduced set of \emph{canonical
  labelings}. In $\R^2$ these are produced by listing all points in
counterclockwise order from some point on the convex hull of $P$. 
In $\R^d$, labelings are generated by convex hull \emph{flags}.
Thus there are at most $O(h) = O(n^{\lfloor d/2 \rfloor})$ canonical orderings where
$h$ is the number of flags  on the convex hull of the sets. 
Using this observation, and the fact that their OTR is of length $O(n^d)$, it was shown
in~\cite{goodman_pollack_83_sorting} that the equality of two order types can be determined in $O(hn^d)
= O(n^{\lfloor 3d/2 \rfloor})$ time.
To our
knowledge, that running time has not been improved for arbitrary $d$.
For $\R^3$ an improvement to $O(n^3\log n)$ has been given for points in general position~\cite{ordertypes_EuroCG12}.\\

\paragraph{Automorphisms and canonical labelings.}
The isomorphism problem is naturally connected to the automorphism
problem, which is to determine the group of permutations $\pi$ such
that  $\nabla^P = \nabla^P\circ\pi$.
One common technique to discover automorphisms is through the use of
\emph{canonical labelings}.
A canonical labeling $\rho^*(\nabla^P)$ for an order type with predicate
$\nabla^P$ is a permutation such that 
$\rho^*(\nabla^P) = \rho^*(\nabla^P\circ\pi)$ for any permutation
$\pi$.
One way of producing such a labeling is to pick $\rho^*$  
(possibly among a reduced set, as done by Goodman and Pollack) as the
labeling that produces the representation $E(\nabla^P\circ\rho^*)$
that is lexicographically minimum (abbreviated as ``{\em MinLex}" later on).
Then, the automorphism group of the order type is just the set of
permutations $\rho$ such that $E(\nabla^P\circ\rho)=E(\nabla^P\circ\rho^*)$.
Of course, using a canonical labeling it is easy to solve the
isomorphism problem as it is sufficient to check whether the canonical
representations of the two order types match.

It is worth noting that the canonical labeling problem could potentially be
 harder than that of isomorphism. For instance, in the case of graphs,  finding a canonical labeling is NP-complete.\\

\paragraph{Our results.}
We present the first $O(n^d)$-time algorithm for producing a 
canonical labeling and the automorphism
group of an order type. Consequently the algorithm can also be used to
determine if two
order types are isomorphic. The algorithm works for any $d\geq 2$ and
does not assume general position. 
It uses no other information than what is given by the order type
predicate (used as an oracle), and works for abstract order types, or
acyclic oriented matroids of rank $d{+}1$.
For abstract order types, it was shown by Goodman and
Pollack~\cite{goodman_pollack_83_sorting} that there are
$2^{\Omega(n^2)}$ different abstract order types of dimension
$2$. Combining this with the information theory lower bound, this
implies that our algorithm is optimal in the abstract case, for
$d=2$. 
If the order type is realizable (i.e., the predicate $\nabla$ is
computed from an actual set of points in $\R^2$), the number of
different order types is much smaller. Goodman and
Pollack~\cite{goodman_pollack_86_upper_bounds} showed that the number
of order types on $n$ points is at least $n^{4n+O(n/\log n)}$ and at
most $n^{6n}$. They improved the upper bound later to
$\left(\frac{n}{2}\right)^{4n(1+O(1/\log(n/2)))}$~\cite{goodman_pollack_91_survey}.
Therefore in this case, the information theory lower bound only gives
a bound of  $\Omega(n\log n)$.
Nevertheless, 
Erickson and Seidel
showed~\cite{DBLP:journals/dcg/EricksonS95,DBLP:journals/dcg/EricksonS97}
using an adversary argument that any algorithm solving order type
isomorphism 
using exclusively the $\nabla$ predicate must query the predicate
$\Omega(n^d)$ times, even if the order type is realizable. 
This shows our algorithm is optimal even for realizable order types in
that model.

\section{Preliminaries}\label{sec:preliminaries}

This section provides a brief and informal overview of the different abstractions for point configurations used in the literature. \\

\paragraph{Euclidean and oriented projective geometry.}
Computational geometers traditionally manipulate points in a Euclidean
plane. When it is necessary or convenient to consider points
at infinity, the projective plane is defined by adding a line at
infinity. More formally, the projective plane is produced by
adding an extra coordinate $z$ to the Euclidean plane $xy$. 
The usual Euclidean plane coincides with the plane $z=1$, that is,
points of Euclidean coordinates $(x,y)$ become $(x,y,1)$. Any point
$(x,y,z)$ is considered to be represented by $(ax,ay,az)$ for all
$a \neq 0$ as well. In other words all points on a line through
the origin $(0,0,0)$ correspond to the same \emph{projective point}.
Thus a projective point can be visualized as a pair of antipodal points on
the sphere $x^2+y^2+z^2=1$. Points at infinity then correspond to
points on the great circle at the intersection between the sphere and
the plane $z=0$. 

However, this transformation from the Euclidean plane to the
projective plane does not preserve the notion of orientation for a
triple of points, as a line through two points does not disconnect the
projective plane. In order to preserve the orientation information,
one could use \emph{oriented projective geometry}~\cite{DBLP:tibkat_025890042}
where points
$(x,y,z)$ are only equivalent to points $(ax,ay,az)$ for $a>0$. 
Thus in the oriented projective plane, a point can be seen as a
\emph{single} point on a sphere $x^2+y^2+z^2=1$. 

Another convenient view of the oriented
projective plane is obtained by gluing two Euclidean planes (at $z=-1$
and $z=1$) using a line at infinity. Every point not at infinity is
then either positive 
or
 negative. A finite collection of points in
the oriented projective plane can
then be represented as a collection of signed points in the plane
$z=1$ by taking the reflection of the negative points through the
origin (the sphere can be rotated to ensure no points lie in the plane
$z=0$).
Note that every triple of points in the oriented projective plane has
a well-defined orientation. In the last representation, the
orientation of a triple of signed points can be computed by multiplying
the unsigned orientation of the points by the signs of the three
points. \\

\paragraph{Projective duality.}
One can gain
 much
 insight into the combinatorial structure of a
discrete
  point set by using \emph{projective duality}. In the
projective plane, the dual of a projective point $(a,b,c)$ is the
plane (called a \emph{projective line}) through the origin with normal
vector $(a,b,c)$, i.e., $ax+by+cz=0$, 
or the great circle where this plane intersects the unit sphere. In
the oriented projective plane, the dual of an oriented point is the
halfspace $ax+by+cz\geq 0$ or the corresponding hemisphere. 
It can 
be verified that this duality transform preserves
incidence between a point and a projective line, and containment between a point
and a halfplane. 

Back in the Euclidean plane $z=1$, this duality transform corresponds
to what is traditionally called the \emph{polar dual}. The point
$(a,b)$ maps to the line $ax+by = -1$ or the halfplane $ax+by \geq -1$
which contains the origin. A set $S$ of points in the Euclidean plane then maps to a
collection $H$ of halfplanes all containing the origin. 
On the sphere, these are hemispheres all containing
the pole $(0,0,1)$ and the convex
hull of $S$ is dual to the intersection of these hemispheres, that is,
the set of hemispheres containing $S$ is dual to the set of points
contained in all the dual hemispheres. 
The arrangement of circles bounding the hemispheres provides some
further information. Each cell $c$ of the corresponding arrangement is
contained in a specific set $H_c\subseteq H$ of hemispheres and is
dual to the hemispheres containing the corresponding set $S_c$ of
points exactly. 
In the oriented projective plane, a negative point will
dualize to a halfplane not containing the origin.
As a set of arbitrary halfplanes or hemispheres is not guaranteed to
have a common intersection,  a set of signed points or a set of points in
the oriented projective plane is not guaranteed to have a bounded convex
hull. In fact, a set of points on the sphere has a bounded convex hull
if and only if all points are strictly contained in a hemisphere. Otherwise the
convex hull is the entire oriented projective plane (or a line if all
points are on the same projective line).
Note however that every triple of hemispheres in general position has
a non-empty intersection, and that the orientation of a triple of
hemispheres can be inferred from the order of appearance of their
boundaries along their intersection.\\

\paragraph{Oriented matroids.}
For a collection $E$ of
 oriented great circles on a sphere $S$ and any
point $q$ on $S$, one 
can write a sign vector indicating for each circle if $q$ is in
the positive ($+$) or negative ($-$) hemisphere, or on the
circle itself ($0$). The resulting vector\footnote{We use the notation
  $\{+,-,0\}^E$ to mean a vector of length $|E|$, whose elements take
  values in $\{+,-,0\}$ and are indexed by the elements of $E$. We
  assume the elements of $E$ are ordered, and we write $X = (X_{e_1},X_{e_2},...)$
  to list the values of a vector in the order of their index set
  $E$. The signed vector $X$ is also interpreted as a signed subset of
  $E$, or a pair of disjoint subsets of $E$: $X=(X^+,X^-)$,
  $X^\sigma=\{e|X_e=\sigma\}$ for $\sigma\in\{+,-,0\}$. We write
  $z(X) = X^0$ and $\underline{X} = X^+ \cup X^-$. 
  Set operations can be used, e.g. if $F\subseteq E$, 
  $X\setminus F = (X^+\setminus F, X^-\setminus F)$, 
  $X|_F = X\setminus(E\setminus F)$.
}   
in $\{+,-,0\}^E$ is called a
\emph{covector}. Let $\L$ be the collection of 
all  such 
covectors for
all points on $S$,  along with the
 the vector
 $\orig=(0,0,\ldots,0)$.
Define the composition operator between sign vectors
  \[(X \circ Y)_e = 
  \begin{cases}
    X_e & \text{if } X_e \neq 0 \\
    Y_e & \text{if } X_e = 0
  \end{cases}
  \text{  }\forall e\in E.
  \]

\noindent The collection $\L$ has several interesting properties:
\begin{description}
\item[(CV0)] $\orig\in\L$.
\item[(CV1)] If $X\in\L$ then $-X\in\L$.
\item[(CV2)] If $X,Y\in\L$ then $X\circ Y\in\L$.
\item[(CV3)] If $X,Y\in\L$, and there exists $e\in E$ such that $\{X_e,Y_e\} = \{+,-\}$. 
Then there is a $Z\in\L$ where $Z_e=0$, and for all $f\in E$ such
  that $\{X_f,Y_f\} \neq \{+,-\}$, $Z_f = (X \circ Y)_f$
\end{description}
(CV1) says every point has an opposite point on the sphere. (CV2) shows
what would happen if you moved by a tiny amount from the point
defining $X$ in the direction of the point defining $Y$. For (CV3), if
points $p$ defining $X$ and $q$ defining $Y$ are separated by 
a projective line $\ell$ then $Z$ would be the covector of the
intersection of the segment $pq$ and $\ell$.

In general, any collection $\L$ that satisfies (CV0-3) defines an
\emph{oriented matroid}. Define the partial order $(\L,\leq)$ where
$X\leq Y$ if $X_e = Y_e$ whenever $X_e\neq 0$. The rank of an oriented
matroid is the length of the longest chain in that partial order,
minus 1. 
In the case of our arrangement of circles, the rank of the associated 
oriented matroid is 3 ($\orig < \text{vertex} < \text{edge} < {face}$).
The \emph{cocircuits} $\C$ of the oriented matroid is the set
of minimal elements in $\L-\{\orig\}$ (in this case, the
vertices of the arrangement). Given a collection $\C$ of cocircuits,
the corresponding set of covectors $\L$ can be reconstructed by
successive compositions (e.g., $C_1\circ C_2 \circ \ldots \circ C_k$) of elements of $\C$.

An oriented matroid is \emph{acyclic} if it contains the
covector $(+,+,\ldots,+)$. This corresponds to the property of all
positive hemispheres having a nonempty intersection, i.e.  the
corresponding set of points in the oriented projective plane has a
bounded convex hull.
An element $e\in E$ is an extreme element of the oriented matroid if there is a
covector $X\in\L$ where $X_e$ is the only positive 
element 
(see~\cite{Bjorner:Oriented}, convexity Proposition 1.6).
 For example if $\L$ represents the
dual arrangement of a planar point set, the covector $X$
represents a halfplane containing only point $e$, and thus $e$ is an
extreme point.   \\

\paragraph{Pseudo-hemispheres and the topological representation theorem.}
Sets of lines in the plane generalize to \emph{pseudolines}, a
collection of topological lines that pairwise intersect and cross exactly
once. In the projective plane, an arrangement of circles
can be generalized to an arrangement of \emph{pseudocircles}, a collection of
closed curves,
every pair of which
intersects and properly crosses 
exactly twice. In the oriented projective plane, each
pseudocircle is given an orientation and defines a positive and
negative \emph{pseudohemisphere}. 

All notions mentioned above generalize to $d$-dimensional spaces. The
generalization of Euclidean, projective, oriented projective spaces
and projective duality is straightforward. The generalization to
pseudospheres and pseudohemispheres requires a bit more care; for the
exact definition see, e.g.~\cite{Bjorner:Oriented} or~\cite{go-hdcg-04}. 
Sign vectors generalize as well and the covectors generated by an
arrangement of oriented pseudospheres on the $d$-sphere define an
oriented matroid (i.e., they satisfy (CV0-3)) of rank $d{+}1$. 
A surprising fact is that the converse
is true: any oriented matroid of rank $d{+}1$ without loops\footnote{A
  loop is an element whose sign is $0$ in every covector in $\L$} can
be realized as a set of oriented pseudospheres on a $d$-sphere 
(\cite{DBLP:journals/jct/FolkmanL78}, Topological representation Thm. 5.2.1, p.
233).  \\

\paragraph{Chirotopes.}
A proper axiomatization generalizing the order type predicate $\nabla$
mentioned in the introduction is provided by the notion of
\emph{chirotope}\footnote{To stay in line with the oriented matroid
  literature, we use the symbol $\chi$ to denote a chirotope in this
  section. In the subsequent sections, we will use $\nabla$ to mean an
abstract order type or an acyclic oriented matroid.}. 
We state one of its several equivalent definitions
for completeness although we will not be using it directly.
A \emph{chirotope} (\cite{Bjorner:Oriented}, p.128) of rank $d{+}1$ for a collection $E$ of $n$ elements is
a non-zero alternating\footnote{A map is alternating if swapping two
  of its arguments negates its value} map 
$\chi: E^{d+1} \rightarrow \{+,-,0\}$ 
satisfying:

\vspace{3mm}
{\noindent\bf(B2$'$)}  For all $x_1,\ldots,x_{d+1},y_1,\ldots,y_{d+1}\in E$  such that
$\chi(x_1,\ldots,x_{d+1}) \chi(y_1,\ldots,y_{d+1})\neq 0$,\\
there is an $i \in \{1,\ldots,{d+1}\}$ such that\\
$\chi(y_i,x_2,\ldots,x_{d+1}) \chi(y_1,\ldots,y_{i-1},x_1,y_{i+1},\ldots,y_{d+1}) =
 \chi(x_1,\ldots,x_{d+1}) \chi(y_1,\ldots,y_{d+1})$.
\vspace{3mm}

\noindent Interestingly, chirotopes are just another possible
representation of oriented matroids, as shown by the following theorem
(Chirotope/Cocircuit translation, Thm.6.2.3,
p.138,~\cite{go-hdcg-04}).

\begin{theorem}[Chirotope/Cocircuit~\cite{go-hdcg-04}]
For each chirotope $\chi$ of rank $d{+}1$, the set
$\C(\chi) = \{(\chi(\lambda,1),\ldots,\chi(\lambda,n)) | \lambda\in E^{d}\}$
is the set of circuits of an oriented matroid of rank $d+1$. Conversely for every
oriented matroid with cocircuits $\C$, there exist a pair of
chirotopes $\{\chi,-\chi\}$ such that $\C(\chi) = \C(-\chi) = \C$. 
\end{theorem} 

\paragraph{Algorithms.}
Edelsbrunner, O'Rourke, and Seidel~\cite{journals/siamcomp/EdelsbrunnerOS86} 
described an algorithm to construct the cell complex of a hyperplane
arrangement in $\R^d$ in time $O(n^d)$. In their conclusion, they
mentioned that their algorithm applies to arrangements of
pseudohyperplanes as well, provided they are computationally simple.
 A careful review of the algorithm reveals that in fact the only
primitive necessary to run the algorithm is to determine whether a
$1$-face of the cell complex (i.e., an edge) is intersected by a
pseudohyperplane. Since the 1-face is defined by $d{+}1$
pseudohyperplanes ($d{-}1$ define the supporting $1$-flat, and the
remaining two delimit the segment), the answer to this primitive can
be computed in $O(1)$ time using chirotope queries, by constructing
the set of cocircuits  of the arrangement of $d{+}2$ pseudohyperplanes
involved. 
The algorithm also needs $d$ pseudohyperplanes in general position in
order to start. For this, pick one arbitrary point and then for every 
${n{-}1}\choose d$ choices of the remaining $d$ points, check the orientation
of the resulting $(d{+}1)$-tuple until a non-zero set is found (which
will happen unless the chirotope predicate is identically 
zero).  
Then iteratively insert
new pseudohyperplanes, updating
the face complex where intersected by the new pseudohyperplane. The
Zone theorem~\cite{Edelsbrunner:1993:ZTH:153794.153828} shows the number of affected faces is $O(n)$.
In the process of constructing the arrangement, the algorithm will
also identify all duplicate copies of elements. We will assume
henceforth that the oriented matroids we consider contain no duplicate
elements.
Note that the algorithm relies on the zone theorem whose
original proof in $\R^d$ was flawed. A new proof however was published
several years later by Edelsbrunner, Seidel, and Sharir~\cite{Edelsbrunner:1993:ZTH:153794.153828}. The new proof also generalizes to pseudohyperplanes.

Once the pseudohyperplane arrangement is constructed, it is
straightforward to determine if the oriented matroid is acyclic
(i.e. if the corresponding abstract order type has a convex hull), in
$O(n^d)$ time, by verifying if there is a cell with covector $(+,\ldots,+)$.  
If it is, in the same running time we can extract from it the convex hull, or all convex
layers (which are constructed iteratively by computing the convex
hull and removing its vertices from $S$). \\

\paragraph{Minors and radial ordering.}
Consider an oriented matroid with covectors $\L\in\{+,-,0\}^E$, and
let $A\subseteq E$ be a nonempty subset of $E$. The
\emph{deletion}
$$\L\setminus A = \{X\setminus A | X\in\L\} 
\subseteq \{+,-,0\}^{E\setminus A}$$
and the \emph{contraction}
$$\L/A = \{X\in\L | A\subseteq X^0\}
\subseteq \{+,-,0\}^{E\setminus A}$$
\noindent each define the set of covectors of another matroid 
(see~\cite{Bjorner:Oriented}, L4.1.8,  p.165). 
When viewing the
oriented matroid as an arrangement of oriented pseudospheres, the set
of pseudospheres in $A$ intersects in a lower dimensional pseudosphere
$S_A$, and $\L/A$ corresponds to the arrangement of the
pseudospheres in $E\setminus A$ on the surface of $S_A$.
The deletion $\L\setminus A$ just corresponds to the deletion of the
hyperspheres in $A$ from the arrangement.
When $A$ contains only one element $e$, we write 
$\L\setminus e = \L\setminus\{e\}$, and $\L/e = \L/\{e\}$.
The following fact will be used by our algorithm:
\begin{proposition}[3.4.8,  p.123 in~\cite{Bjorner:Oriented}]
If the matroid with cocircuits $\L\subseteq \{+,-,0\}^E$ is acyclic
and $e\in E$, then $\L/{e}$ is acyclic if and only if $e$ is
an extreme element.
\end{proposition}

The contraction for a chirotope $\chi$ of rank $d{+}1$, assuming
$|A|\leq d$ is the restriction to fixing the $|A|$ first arguments of
$\chi$ to the elements of $A$, that is, 
$$\chi_A(x_1,\ldots,x_{d+1-|A|}) = \chi(A,x_1,\ldots,x_{d+1-|A|}).$$
Again, if $A=\{e\}$ we write $\chi_e = \chi_{\{e\}}$.
For example, if $\chi$ is the order type of a set of points in $\R^d$,
then the restriction $\chi_{\{q\}}$ corresponds to the central projection of
all points on a sphere around $q$.

Suppose $E$ is a set of  planar points, $\chi$ is
its order type and $\L\subseteq \{+,-,0\}^{|E|}$ is the set of
covectors. An oriented line $\ell$ rotating about a point $q\in E$ will meet all
other points of $E$ in a cyclic fashion (some simultaneously). 
In a full rotation, $\ell$
will meet each point exactly twice, once on its positive side and once
on its negative side. The contraction $\L/q$ is an oriented
matroid of rank 2. It corresponds to an arrangement of oriented
0-spheres on a $1$-sphere (i.e., a circle). That is, each element is
represented by two points identifying a semi-circle on the
circle. Walking clockwise in the positive semi-circle corresponding to
some point $p$, every other element will appear exactly once, either
entering its positive or its negative semi-circle. Walking
clockwise in the negative semi-circle of $p$, the other elements of
$E\setminus\{p,q\}$ are encountered in the same order, but the sign of the semi-circle
entered is reversed.

More generally, suppose $\L\subseteq \{+,-,0\}^{|E|}$ is the set of covectors
of an oriented matroid of rank $d+1$ and $\chi$ is the corresponding
chirotope.
Given a subset $A\subseteq E$ of 
$d$ elements in general position
(i.e., there is a $e\in E$ such that $\chi(A,e)\neq 0$),
the contraction $\L/ A$ is an oriented matroid of rank $2$, and
induces a double (signed) cyclic ordering (with ties) of all other elements of
$E$. If $\L$ corresponds to a set of points in $\R^d$, then this is
the order in which the points of $E$ are  swept by a hyperplane
rotating about the points of $A$.   If the pseudosphere
arrangement of $\L$ has been precomputed, then the cyclic ordering of
$\L/A$ can be found in $O(n)$ time by a simple walk in the
associated data structure.   \\

\paragraph{Flags.}
The method of Goodman and Pollack~\cite{goodman_pollack_83_sorting} for order type isomorphism, as well
as the one presented here, starts by defining a small set of good
candidate canonical orderings of the point set.
As discussed above, contracting an oriented matroid to one of rank 2
produces a cyclic ordering, however that ordering might have ties
(e.g. points can be swept simultaneously by a line in $\R^2$), and
we haven't determined where the canonical ordering should start.
For this, we translate the \emph{face flags} used by Goodman and Pollack to
the language of oriented matroids.

Let $\nabla$ be an abstract order type of dimension $d$, or the
chirotope of an acyclic oriented matroid of rank
$d{+}1$. Let $\L$ be the set of covectors of that 
oriented 
matroid.
Assume the corresponding arrangement $\A$ of pseudohyperspheres has
been constructed in  $O(n^d)$ time using the algorithm above.

The $(d{-}1)$-{\em facets} of the convex hull of the (abstract)
order type $\nabla$ are the vertices of the $(+,\ldots,+)$ face of the
arrangement $\A$.
In general, the $(d{-}1{-}i)$-faces of the convex hull are the $i$-faces
of the $(+,\ldots,+)$ face of $\A$.

A sequence of covectors $\phi = (X^{(1)},X^{(2)},\ldots,X^{(d-1)})$ is a
face-flag if $(+,\ldots,+)=X^{(0)}>X^{(1)}>\ldots>X^{(d-1)}>X^{(d)}>\orig$
is a maximal chain for some $X^{(d)}$ 
in the covector partial order $(\L,\leq)$.
Each $X^{(i)}$ is an $i{-}1$-face of the convex hull.  
For each $1\leq i\leq d-1$, let $e_i$ be some element in
$z(X^{(i)})\setminus z(X^{(i-1)})$. Such an element always exists by
the strict inequalities in the chain.
Goodman and Pollack~\cite{goodman_pollack_83_sorting} give an upper bound of
$O(n^{\lfloor d/2\rfloor})$ on the number of face-flags for the convex
hull of a set of points in $\R^d$. The proof essentially applies the
upper bound theorem, which is valid for oriented matroids of rank
$d{+}1$, so the same bound applies to the general case.

We will now define an ordering $\pi_\phi:[n]\rightarrow E$ on the
elements of $E$ determined by the face-flag $\phi$.
The covector $X^{(d-1)}$ is a (1-dimensional) circle $C$ in $\A$.
The contraction $\L' = \L/z(X^{(d-1)})$ is an acyclic oriented matroid 
(see~\cite{Bjorner:Oriented}, Proposition
9.1.2,  p.378) of rank 2. Thus, it is equivalent to an
arrangement of 0-spheres on $C$. The arrangement for $\L'$ can
be reconstructed in $O(n)$ time using a precomputed
arrangement for $\L$.
Let $Y$ and $Z \in \L'$ be the two covectors of the two vertices
bounding the positive face $(+,\ldots,+)$ on $C$.
Pick two elements $e_Y\in z(Y)$ and $e_Z\in z(Z)$. The positive
direction along the circle is defined using the sign of
$\nabla(e_1,\ldots,e_{d-1},e_Y,e_Z)$. Assume w.l.o.g. that this sign is
$+$ (otherwise swap $Y$ and $Z$). 

Walking in the positive direction along $C$ starting from the facet
$(+,\ldots,+)$, we encounter vertices with covectors
$Y=Y^{(1)},\ldots,Y^{(k)},-Y{(1)},\ldots,-Y^{(k)}$. Let 
$E_i = z(Y^{(i)})\cup z(X^{(d-1)})$ for $i=1,\ldots,k$. 
Each deletion $\L^{(i)} = \L\setminus(E\setminus E_i)$ is an acyclic oriented
matroid of rank $d$ and contains the flag
$\phi'=(X^{(1)},X^{(2)},\ldots,X^{(d-2)})$.
In order to compute the order $\pi_\phi$ for $E$ in $\L$, we
recursively compute the order using flag $\phi'$ for $E_i$ in
$\L^{(i)}$. The resulting order for $E$ is obtained by listing the
elements of each $E_i$, $i=1,\ldots,k$, omitting the elements from
$z(X^{(d-1)})$ for $i\geq 2$.

\begin{theorem}
Assume the arrangement $\A$ of an acyclic oriented matroid $(E,\L)$ has been
precomputed. 
Then, given a face-flag $\phi$, it is possible in $O(n)$ time to
produce an order $\pi_\phi$ of the elements of $E$ that only depends
on $\L$ and $\phi$. 
\end{theorem}

\section{2D}
The order type of a planar set $P$ of $n$ points
is characterized by a predicate $\nabla^P(i,j,k) = \nabla(p_i,p_j,p_k)$, $i,j,k\in [n]$ whose
sign is $0$, $-$, or $+$, depending on whether the ordered triple
$(p_i,p_j,p_k)$ is collinear, clockwise, or counterclockwise,
respectively.
Our algorithm will work for any predicate $\nabla^P$ that
satisfies the axioms of CC-systems~\cite{DBLP:books/sp/Knuth92}, or equivalently acyclic chirotopes
of rank 3.
 We will assume that duplicate points have been
 identified and  $P$ contains
 distinct points, not all collinear. \\

\paragraph{Spiral Labelings.}
Let $c_1, c_2, \ldots, c_m$ be the convex layers of $P$, where the
successive layers are constructed iteratively by computing the convex
hull of $P$ (including all points on the edges of the convex hull) and
removing the corresponding points from $P$.  Note that all convex
layers except possibly $c_m$ contain at least three points. The case
where $c_m$ contains exactly one point will be treated with special care below.
Using a semi-dynamic convex hull data structure, 
Chazelle~\cite{chazelle_85_convex_layers} showed
how to compute all convex layers in $O(n \log n)$ time. Although his
algorithm does not exclusively use  order type information (for
instance, it compares the $x$-coordinates of input points), a careful reading of
the article reveals that the algorithm can be modified to use only the
order type predicate. For instance, the $x$-coordinate order can be
replaced by the counterclockwise order of the points seen from an
arbitrary point on the convex hull of $P$. 
For the reader unwilling to delve into the details of that algorithm,
note that 
the convex layers can be constructed 
via a much simpler $O(n^2)$-time algorithm, 
by repeated use of the Jarvis March~\cite{Jarvis197318},
or by constructing the dual arrangement as explained in
Section~\ref{sec:preliminaries}.  
Although not optimal, this  running time will be sufficient for our purpose.

Then for each vertex $p$ on convex layer $c_j$,  for $j<m$, construct an edge
from $p$ 
to the first vertex $\tau(p)$ encountered on the counterclockwise tangent to
$c_{j+1}$ nested within. 
All such edges can be found in $O(n)$ time by
walking in parallel counterclockwise along $c_j$ and $c_{j+1}$, for
all $1\leq j < m$. 

For any convex hull vertex $p$, define the \emph{spiral
  labeling} $\rho_p$ as follows.
Traverse the convex hull in counterclockwise order starting at $p$, 
follow the tangent from the last encountered hull point to the next convex layer, and repeat. 
The first node encountered 
in
 the spiral labeling on each layer is
called a \emph{knob}.
Just like
the canonical orderings of Goodman and Pollack,  $\rho_p$ only depends on the
order type and
on
 the choice of
 $p$.

Within any layer that contains at least two points,
let $v'$ be the point counterclockwise to any given point $v$.
The oriented line through
$vv'$
 divides the first layer $c_1$ (or any layer that contains $v$) into two
nonempty subchains. Define 
$s(v)$   
as the most counterclockwise point $q$
on $c_1$ for which 
$\nabla(v',v,q)=-$.
If $v$ is 
 on $c_1$, set
$s(v)$
 to the point clockwise to 
 $v$
on $c_1$.
The order type and $v$ uniquely determine $s(v)$.

Among all layers containing at least two points, 
suppose $c_j$   contains the minimum number of points, and let $k = |c_j|$
 be the number of vertices in that layer.
 Therefore, $k(m{-}1)+1 \leq n$. 
Let $K = \{s(p)|p\in c_j\}$ be a set of at most $k$ \emph{keypoints} on the
convex hull $c_1$. 
The set $K$ depends only on the order
type of $P$. This immediately suggests a slight improvement over
Goodman and Pollack's restriction to canonical labelings: it is
sufficient to look only at labelings (e.g., spiral labelings)
generated by keypoints. 
Combining these $k$ labelings with any $O(n^2)$ size OTR (such as the
one defined by Goodman and Pollack~\cite{goodman_pollack_83_sorting}), 
we thereby obtain a 
$O(kn^2)$-time algorithm for testing order type isomorphism or finding a
MinLex labeling. We will improve this further in the next sections. \\

\paragraph{The Universal Standard Spiral Representation (USSR).}
Given a spiral labeling $\rho$, we describe here an OTR of size
$O(n^2)$. 
Although this is 
not the OTR on which the MinLex labeling
will be based,
but
 it is a first step towards building such an OTR.
For convenience, let $p_i = \rho(i)$ for $i = 1,\ldots,n$.

The \emph{Universal Standard Spiral Representation (USSR)} will be
structured as $n$ blocks, $B_1,\ldots,B_n$, one for each point, where 
successive blocks are separated using a special semicolon ';'. Assuming point $p_i$ is
on 
$c_j$,
block $B_i$ will represent the orientation of $p_i$ with
every pair of points 
in 
$A := c_1 \cup \ldots \cup c_j$. 
This will clearly
constitute an OTR 
because the orientation of any  triple
of points will be encoded on the block of the point(s) on the
deepest layer among those three. 
If the deepest layer $c_m$ contains only
one point, that point always has the last label in any spiral
labeling, that is, it is $p_n$. In that case, 
the last block 
$B_n$ is
called the \emph{East Block}.

Each block $B_i$ will list all points of $A$ in
radial order. Special care is taken in order to 
handle degeneracies and ensure that the representation only depends on
the order type, the labeling, and $p_i$. First, separate all points
into sets 
$A^\sigma = \{q\in A | \nabla(p_i,s(p_i),q) = \sigma\}$
for $\sigma = -,+,0$. The {\em  radial} order is 
that in which a line
passing through $p_i$ rotating in counterclockwise
direction encounters
the points of $A$. Groups of points collinear with $p_i$ will be
equal in the order. The order will start with the set $A^0$. Then
the order for points in $A^+\cup A^-$ can be found using a standard sorting
algorithm with $\gamma\nabla(p_i,q,r)$ as a comparison operator where
$\gamma$ is $+1$ if $q$ and $r$ are both in $A^+$ or both in $A^-$,
and $-1$ if they are in different sets.  
The order is  encoded by writing the labels of the points in their
order, preceding each point in $A^-$ by the symbol ``$-$'' and each
point in $A^+$ by ``+'', and collecting
groups of ``equal" (collinear) points in parentheses. 
For points in $A^0$, precede each point by ``$-$'' if it is before $p_i$
in the direction $p_is(p_i)$, and by ``+'' otherwise.
The collinear
points are listed in increasing order in the direction $p_is(p_i)$ for
$A^0$, and in the direction going from points in $A^-$ (or $p_i$ if
there is none) to points in $A^+$ (or $p_i$ if there is none). 
Finally, if $c_m$ contains only one point, $p_n$, then replace
$s(p_n)$ (which is not defined) by $p_1$ in the above description for
block $B_n$.

The computation of the radial order in each block takes $O(n\log n)$
time, and so for all blocks $O(n^2 \log n)$. 
However $O(n^2)$ can be
achieved for all blocks by building the dual arrangement as described in
Section~\ref{sec:preliminaries} (or see, for
instance,~\cite{DBLP:journals/dm/GilSW92}).

Although the USSR produces a string of $O(n^2)$ size in $O(n^2)$
time, the string could change significantly if the spiral labeling of a
different keypoint were to be used. First, the blocks would have to be
reordered according to the new labeling, then although the order of
the points would remain the same within each block except possibly the
last one, all the labels for the points listed in that block would
change. Finally in the case where $c_m$ is of size 1, the last block would
have to be recomputed. Performing these modifications explicitly
would take time $O(n^2)$ for each of the labelings $\rho_p$ for
$p\in K$ so we will try to simplify the representation to allow for
fast lexicographic comparison while making an
explicit reconstruction unnecessary.\\

\paragraph{The Dumbed Down Representation (DDR)}
 is identical to the USSR except that
each vertex label is replaced by the level number on which that vertex lies.
As the level of a point does not depend on a specific labeling, but only
on the order type, this makes an implicit computation of the DDR for a different labeling
much easier. First notice that, although the blocks are reordered
according to the new labeling, the content of each block (except
possibly the last one) remains identical since it no longer depends on
the labeling. In the case when $c_m$ contains only the point $p_n$, the
block $B_n$ is recomputed because the starting point $\rho(1)$ for
the radial ordering
changes.

Therefore, after 
computing the first DDR for one spiral labeling, each
subsequent DDR can be computed implicitly for each
labeling
 in $O(n)$
time by reordering the blocks and possibly recomputing the last one.
As this would have to be
 done for each labeling $\rho_p$, $p\in K$,
the total construction cost is  $O(kn) = O(n^2)$.
It remains to show how to 
 find the lexicographically smallest
of the $O(k)$ DDR strings
in quadratic time. Since each of them has length $O(n^2)$, we
need to find a way to compare DDR strings without reconstructing them
explicitly. 

After constructing all blocks of the first DDR, build a
trie containing all the blocks, in $O(n^2)$ time (linear in the total
length of the strings). 
A  simple
 in-order walk in the
trie
 will reveal
the lexicographic order of the $n$ blocks, and if any are
identical. Assign to each distinct block a new letter with an order
that matches
 the lexicographic order of the blocks. 
 Rewriting
the DDR using these new letters for all blocks except the last one
which is kept intact, we obtain a string of length $O(n)$ on an
alphabet of length $O(n)$.  This substitution preserves the
lexicographic order, as that order is determined by the first mismatch
between two DDR strings. If the mismatch occurs in one of the $n{-}1$
first blocks, then the new letter for that block will be the first
mismatch in both compressed strings.

Using the same alphabet, we can 
write down
 a compressed version of the DDR for each of the other labelings
in $O(n)$ time: Shuffle the letters of the $n{-}1$ first blocks using
the new labeling, and reconstruct the last block.
Thus the lexicographically smallest DDR can be found in $O(kn)$
time.
This would solve our problem if the DDR was an OTR. However in some
cases the DDR might not contain enough information to recover the
orientation of some triples. To remedy this, we will complement the
DDR with some extra information. \\

\paragraph{The Knob Groups Block (KGB).}
Recall that the knobs of a spiral labeling $\rho$ are the first vertices
encountered by the spiral on each layer.
For each block $B_i$ of the USSR for $\rho$, take note of the position
of the knob of each layer within $B_i$. This is the {\em Knob Group} of
$B_i$. Write down the Knob Groups for $B_1,\ldots,B_n$ consecutively,
separating consecutive $B_i$ by a colon ','. Call the resulting string
the \emph{Knob Groups Block (KGB)}.
The number of knobs recorded for each block is $m$ so the total length
of the KGB is $O(mn)$. 
The knobs for any spiral labeling can be found in $O(m)$ time using
the precomputed tangents $\tau(p)$. Therefore, after computing the
USSR for any spiral labeling, the KGB for any other spiral labeling
can be constructed in $O(mn)$ time.\\

\paragraph{DDR + KGB = USSR.}
For any spiral labeling $\rho$, the labels within a layer $c_j$ 
are drawn from the same set of integers
$[|c_1|+\ldots+|c_{j-1}|+1,  |c_1|+\ldots+|c_j|]$ and are consecutive
along the boundary of $c_j$. Therefore if the knob in each layer is
known, as well as the vertices of each layer in counterclockwise
order, then the spiral labeling can be reconstructed.
In each block $B_i$ of a DDR, the counterclockwise ordering of the
vertices in layer $j$ is exactly the order in which $+j$ appears
followed by the order in which $-j$ appears. Therefore, from block
$B_i$ of a DDR, and using the KGB, the corresponding block of the USSR
can be reconstructed: For each layer number $j$, use the KGB to find
the occurrence of $\gamma j$ that corresponds to the knob on layer
$j$, where $\gamma$ is either $+$ or $-$. Replace $j$ by 
$|c_1|+\ldots+|c_{j-1}|+1$, then replace all subsequent occurrences of
$\gamma j$ sequentially by incrementing the label. 
Let $\bar\gamma$ be the opposite sign of $\gamma$. Starting from the
beginning of $B_i$, continue by replacing successive occurrences of 
$\bar\gamma j$, and finally starting again from the beginning of $B_j$
replace the remaining occurrences of $\gamma j$ until returning to the
knob. 

We have  shown that using a DDR and a KGB, we can reconstruct the
corresponding USSR. This implies that the concatenation of the DDR and
the KGB is an OTR. Using the compressed DDR, the
total length of this OTR is $O(mn)$, and the DDR+KGB OTR can be
constructed for each of the $O(k)$ spiral labelings $\rho_p$, $p\in K$
in $O(mn)$ time. Therefore the total construction time, and the time
to pick the labeling that produces the MinLex OTR is 
$O(kmn)$.  Recalling that $km=O(n)$, we obtain the desired bound of $O(n^2)$.

\section{$\R^d$}

An order type of a point set $P$ in $\R^d$ is characterized by a predicate
$\nabla(p_0,p_1,\ldots,p_d)$. Goodman and Pollack (see~\cite{goodman_pollack_83_sorting}, Lemma 1.7) showed that for any point $q$ on the convex hull of
$P$, the predicate 
$\nabla_q(p_0,\ldots,p_{d-1}) = \nabla(p_0,\ldots,p_{d-1},q)$
characterizes the order type of a point set in $\R^{d-1}$.

In $\R^d$, a \emph{face-flag}~\cite{goodman_pollack_83_sorting} is a sequence 
$\phi = (\phi_0,\phi_1,\phi_2,\ldots,\phi_{d-1})$ 
of faces where $\phi_i$ is of dimension $i$ and $\phi_i$ is a face of
$\phi_{i+1}$. Goodman and Pollack showed how any flag induces a
labeling $\rho_\phi$ of the point set, in a manner
very similar to what was described above when defining the USSR.
\begin{theorem}
\label{theo:Rd}
Given an order type predicate $\nabla$ for an abstract order type in
dimension $d$ (or an acyclic oriented matroid of rank $d{+}1$), there
is an algorithm that in time $O(n^d)$ determines the automorphism
group of $\nabla$. 
It outputs a maximal set of canonical labelings $\Psi(\nabla) =
(\rho_1,\ldots,\rho_k)$ such that 
$\nabla\circ\rho_i=\nabla\circ\rho_j$ for $i,j\in\{1,\ldots,k\}$. 
\end{theorem}

\begin{corollary}
Given an order type predicate $\nabla$ for an abstract order type in
dimension $d$ (or an acyclic oriented matroid of rank $d{+}1$), there
is an algorithm that in time $O(n^d)$ computes a canonical
representation $OTR(\nabla)$ of size $O(n^d)$, and
$OTR(\nabla\circ\rho)$ is the same for all $\rho\in\Psi(\nabla)$.
\end{corollary}

As in the 2D case, construct the dual arrangement for $\nabla$ and the
convex layers $c_1,\ldots,c_m$, with
the small distinction that this time, layer $c_i$ contains only the
extremal points of the point set when the points of previous layers
have been removed. That is, points in the interior of facets of the
convex hull are not included in the layer. 
The arrangement and all convex layers can be computed in $O(n^d)$ time
as explained in Section~\ref{sec:preliminaries}.
Let $\Phi$ be the set of face flags for $c_1$. Then $|\Phi| = O(n^{\lfloor d/2\rfloor})$.

Let $\nabla(p_0,\ldots,p_d)$ be the $d$-dimensional order type predicate.
As mentioned, for any point $q$ on $c_1$, restricting $\nabla$ to that
point produces an order type in $d{-}1$
dimensions, $\nabla_{q}(p_0,\ldots,p_{d-1})=\nabla(p_0,\ldots,p_{d-1},q)$.
By induction, a canonical representation for that order type, as well the associated canonical labeling(s), can be
found in $O(n^{d{-}1})$ time.
For a pair of labelings $\pi$ on $n$ elements and  $\rho$ on $r$
elements of $E$, we define the sequence
$\pi[\rho] = (\pi(\rho^{-1}(1)),\ldots,\pi(\rho^{-1}(r)))$, which
encodes the labeling $\rho$ using $\pi$.
We can now describe our representation for any labeling 
$\pi_\phi$ for $\phi\in\Phi$.

\begin{algorithmic}
\Function{OTR}{$\nabla \circ \pi_\phi$}
\For{$i= 1,\ldots,m$}
\State   (A) Let $\nabla^{(i)} = \nabla\setminus(\bigcup_{j<i}c_i)$ 
  \Comment{$\nabla$ for points of layer $i$ and up.}
\ForAll{ $q\in c_i$ in the order of $\pi_\phi$}
\State      (B) Recursively compute $\Psi(\nabla^{(i)}_q)$.
              \Comment{\parbox{2.1in}{Since $q$ is extremal for $c_i\cup\ldots\cup c_m$, $\nabla^{(i)}_q$ is acyclic.}}
\State      (C) Find the $\rho_{min}\in\Psi(\nabla^{(i)}_q)$ such that $\pi_\phi[\rho_{min}]$ is lexicographically minimum.
\State     Write $\pi_\phi[\rho_{min}]$
\State      Write $OTR(\nabla^{(i)}_q\circ\rho_{min})$
\EndFor
\EndFor
\EndFunction
\end{algorithmic}

\begin{lemma} 
The output of $OTR(\nabla \circ \pi_\phi)$ encodes the order type of
$\nabla$.
\end{lemma}
\begin{proof}
For any $x_1,\ldots,x_{d+1}$, let $i$ be the smallest number such that
$c_i$ contains some point $x_j$. Look up
$OTR(\nabla^{(i)}_{x_j}\circ\rho_{min})$ and the corresponding 
$\pi_\phi[\rho_{min}]$. All points $x_1,\ldots,x_{d+1}$ are in 
$c_i\cup\ldots\cup c_m$, and the label $\rho_{min}(y)$ for 
$y\in c_i\cup\ldots\cup c_m$ is the rank of
$\pi_\phi(y)$ in $\pi_\phi[\rho_{min}]$. 
Therefore, the value of 
$\nabla(x_1,\ldots,x_{d+1}) =
(-1)^{(j-1)}\nabla^{(i)}_{x_j}(x_1\ldots,x_{j-1},x_{j+1},\ldots,x_{d+1})$
can be retrieved recursively from $OTR(\nabla^{(i)}_{x_j}\circ\rho_{min})$.
\qed\end{proof}

We now turn to the size of the string. 
By induction, we observe that
$OTR(\nabla \circ \pi_\phi)$ has size $O(n^d)$ as the double loop runs
exactly once for each of the $n$ points, $\pi_\phi[\rho_{min}]$ is of
size $O(n)$ and $OTR(\nabla^{(i)}_q\circ\rho_{min})$ is of size
$O(n^{d-1})$, by induction.  The automorphism group and a set of
canonical labelings for $\nabla$ can be computed by producing
$OTR(\nabla \circ \pi_\phi)$ for all $\phi\in\Phi$ and outputting the
labelings that correspond to the lexicographically minimum OTR.
However, as $|\Phi| = O(n^{\lfloor d/2\rfloor})$, the running time of
this algorithm would be $O(n^{\lfloor 3d/2\rfloor})$. In fact this is
roughly the same algorithm as the one of Goodman and Pollack~\cite{goodman_pollack_83_sorting}.

In order to speed up this process, we notice that each
$OTR(\nabla^{(i)}_q\circ\rho_{min})$ does not depend on the choice of
$\phi$, and therefore will be the same in all 
$O(n^{\lfloor d/2\rfloor})$ OTRs. Therefore, as a preprocessing step,
we can produce $T_q =OTR(\nabla^{(i)}_q\circ\rho_{min})$ 
for all $q$ in layer $i$ and for all layers, in $O(n^d)$ time, and sort them lexicographically. 
Create a new character for each distinct $T_q$ with the same
ordering. Now for each $\phi$, we can write the compressed OTR,
replacing $T_q$ by the corresponding letter. The compressed OTR is
thus of length $O(n^2)$ and the lexicographic order matches the
order of the uncompressed strings. 
The total cost of finding the lexicographically minimum compressed
OTRs and the corresponding flags $\phi$ is 
$O(n^{\lfloor d/2 \rfloor + 2})$, which is $O(n^d)$ for $d\geq 3$.

We are now left with the delicate task of bounding the time needed to
produce each compressed OTR. Step (A) is implicit and this has no
cost. Step (B) takes time $O(n^{d-1})$, by induction. However, this
step does not depend of $\phi$ and so could be precomputed in a
preprocessing phase, in a total time $O(n^d)$.
Step (C) requires to compare up to $O(n^{\lfloor (d-1)/2\rfloor})$ strings
each of length $O(n)$. This step does depend on $\phi$. 
The total cost for all $\phi$ would then be 
$O(n^{\lfloor d/2\rfloor + \lfloor (d-1)/2\rfloor +2})$, which is
$O(n^{d+1})$.
We will need to work a bit more to remove the extra factor of  $n$.

Pick one arbitrary $\phi_0\in\Phi$. 
When precomputing all $\Psi(\nabla^{(i)}_q)$, store the sequences
$\pi_{\phi_0}[\rho]$ for each $\rho\in \Psi(\nabla^{(i)}_q)$ in a
compressed trie data structure. To look up $\rho_{min}$ given a
particular $\pi_\phi$, walk down the trie, always choosing the
child whose label $s$ minimizes $\pi_\phi(\pi_{\phi_0}^{-1}(s))$.
The size of the alphabet is bounded by $O(n)$, therefore the degree of each node of the trie is $O(n)$. Each string is bounded by $O(n)$ so the height of the trie is $O(n)$. This implies the total cost of a lookup is $O(n^2)$.
This lookup is performed for each point and each
flag $\phi\in\Phi$. Therefore, the total cost is 
$O(n^{\lfloor d/2\rfloor + 3})$, which is $O(n^d)$ for $d{\geq}5$.
For $d=3$ and $4$,
the total degree of the
entire trie is no more than the number of leaves in the trie, that is, 
$O(n^{\lfloor (d-1)/2\rfloor})$ and the cost of a query cannot exceed
that bound. Therefore the total cost for all points and flags is 
$O(n^{\lfloor d/2\rfloor + \lfloor (d-1)/2\rfloor +1})$, which is
$O(n^{d})$ for $d=3$ and $4$ as well. This completes the proof of Theorem~\ref{theo:Rd}.

\section*{Acknowledgements}
This work was initiated at the 2011 Mid-Winter
Workshop on Computational Geometry. The authors thank all participants
for providing a stimulating research environment. We also thank the
anonymous referees for numerous helpful comments.
\bibliography{refs}

\begin{thebibliography}{10}

\bibitem{aichholzer_etal_02}
O.~Aichholzer, F.~Aurenhammer, and H.~Krasser.
\newblock Enumerating order types for small point sets with applications.
\newblock {\em Order}, 19(3):265--281, 2002.

\bibitem{aichholzer_krasser_05_abstract_order_type}
O.~Aichholzer and H.~Krasser.
\newblock Abstract order type extension and new results on the rectilinear
  crossing number.
\newblock In {\em Symposium on Computational Geometry}, pages 91--98, 2005.

\bibitem{ordertypes_EuroCG12}
G.~Aloupis, M.~Dulieu, J.~Iacono, S.~Langerman, O.~\"Ozkan, S.~Ramaswami, and
  S.~Wuhrer.
\newblock Order type invariant labeling and comparison of point sets.
\newblock In {\em Proceedings of the 28th European Workshop on Computational
  Geometry}, pages 213--216, 2012.

\bibitem{Bjorner:Oriented}
A.~Bj{\"o}rner, M.~Las~Vergnas, B.~Sturmfels, N.~White, and G.~M. Ziegler.
\newblock {\em Oriented Matroids}.
\newblock Cambridge University Press, second edition, 1999.

\bibitem{chazelle_85_convex_layers}
B.~Chazelle.
\newblock On the convex layers of a planar point set.
\newblock {\em IEEE Transactions on Information Theory}, 31(4):509--517, 1985.

\bibitem{multiplex}
A.~Dreiding and K.~Wirth.
\newblock The multiplex. {A} classification of finite ordered point sets in
  oriented d-dimensional space.
\newblock {\em Journal of Mathematical Chemistry}, 8:341--352, 1980.

\bibitem{dreiding}
A.~Dress, A.~Dreiding, and H.~Haegi.
\newblock Chirotopes and oriented matroids.
\newblock {\em Bayreuther Mathematische Schriften}, 21:14--68, 1986.

\bibitem{journals/siamcomp/EdelsbrunnerOS86}
H.~Edelsbrunner, J.~O'Rourke, and R.~Seidel.
\newblock Constructing arrangements of lines and hyperplanes with applications.
\newblock {\em SIAM Journal of Computing}, 15(2):341--363, 1986.

\bibitem{Edelsbrunner:1993:ZTH:153794.153828}
H.~Edelsbrunner, R.~Seidel, and M.~Sharir.
\newblock On the zone theorem for hyperplane arrangements.
\newblock {\em SIAM Journal of Computing}, 22(2):418--429, Apr. 1993.

\bibitem{DBLP:journals/dcg/EricksonS95}
J.~Erickson and R.~Seidel.
\newblock Better lower bounds on detecting affine and spherical degeneracies.
\newblock {\em Discrete {\&} Computational Geometry}, 13:41--57, 1995.

\bibitem{DBLP:journals/dcg/EricksonS97}
J.~Erickson and R.~Seidel.
\newblock Erratum to better lower bounds on detecting affine and spherical
  degeneracies.
\newblock {\em Discrete {\&} Computational Geometry}, 18(2):239--240, 1997.

\bibitem{DBLP:journals/jct/FolkmanL78}
J.~Folkman and J.~Lawrence.
\newblock Oriented matroids.
\newblock {\em Journal of Combinatorial Theory, Series B}, 25(2):199--236,
  1978.

\bibitem{DBLP:journals/dm/GilSW92}
J.~Gil, W.~L. Steiger, and A.~Wigderson.
\newblock Geometric medians.
\newblock {\em Discrete Mathematics}, 108(1-3):37--51, 1992.

\bibitem{goodman_pollack_86_upper_bounds}
J.~Goodman and R.~Pollack.
\newblock Upper bounds for configurations and polytopes in ${R}^d$.
\newblock {\em Discrete \& Computational Geometry}, 1:219--227, 1986.

\bibitem{goodman_pollack_91_survey}
J.~Goodman and R.~Pollack.
\newblock The complexity of point configurations.
\newblock {\em Discrete Applied Mathematics}, 31:167--180, 1991.

\bibitem{go-hdcg-04}
J.~E. Goodman and J.~O'Rourke, editors.
\newblock {\em Handbook of Discrete and Computational Geometry}.
\newblock CRC Press LLC, second edition, 2004.

\bibitem{Goodman1980220}
J.~E. Goodman and R.~Pollack.
\newblock On the combinatorial classification of nondegenerate configurations
  in the plane.
\newblock {\em Journal of Combinatorial Theory, Series A}, 29(2):220 -- 235,
  1980.

\bibitem{goodman_pollack_83_sorting}
J.~E. Goodman and R.~Pollack.
\newblock Multidimensional sorting.
\newblock {\em SIAM Journal on Computing}, 12(3):484--507, 1983.

\bibitem{Goodman1984257}
J.~E. Goodman and R.~Pollack.
\newblock Semispaces of configurations, cell complexes of arrangements.
\newblock {\em Journal of Combinatorial Theory, Series A}, 37(3):257 -- 293,
  1984.

\bibitem{Jarvis197318}
R.~Jarvis.
\newblock On the identification of the convex hull of a finite set of points in
  the plane.
\newblock {\em Information Processing Letters}, 2(1):18 -- 21, 1973.

\bibitem{DBLP:books/sp/Knuth92}
D.~E. Knuth.
\newblock {\em Axioms and Hulls}, volume 606 of {\em Lecture Notes in Computer
  Science}.
\newblock Springer, 1992.

\bibitem{Novoa}
L.~G. Novoa.
\newblock On n-ordered sets and order completeness.
\newblock {\em Pacific Journal of Mathematics}, 15(4):1337--–1345, 1965.

\bibitem{Perrin1882}
R.~Perrin.
\newblock Sur le probl\`{e}me des aspects.
\newblock {\em Bulletin de la Soci\'{e}t\'{e} Math\'{e}matique de France},
  10:103--127, 1882.

\bibitem{DBLP:tibkat_025890042}
J.~Stolfi.
\newblock {\em Oriented projective geometry: a framework for geometric
  computations}.
\newblock Academic Press, Boston, MA, 1991.

\end{thebibliography}
\bibliographystyle{abbrv}   
\end{document}